\documentclass[a4paper, 12pt]{article}

\usepackage[utf8]{inputenc}
\usepackage[T1]{fontenc}
\usepackage[a4paper, margin=2.5cm]{geometry}
\usepackage{needspace}
\usepackage{comment}
\usepackage{libertine} \usepackage{inconsolata} 

\usepackage{amsmath, amsthm, amssymb}
\usepackage{graphicx}
\usepackage{enumerate}
\usepackage{authblk}
\usepackage[textsize=small, textwidth=2cm, color=yellow]{todonotes}
\usepackage{thmtools}
\usepackage{thm-restate}
\usepackage[colorlinks=true, citecolor=red]{hyperref}
\usepackage{float}

\renewenvironment{abstract}
{\small\vspace{-1em}
\begin{center}
\bfseries\abstractname\vspace{-.5em}\vspace{0pt}
\end{center}
\list{}{
\setlength{\leftmargin}{0.6in}\setlength{\rightmargin}{\leftmargin}}\item\relax}
{\endlist}

\declaretheorem[name=Theorem, numberwithin=section]{theorem}
\declaretheorem[name=Lemma, sibling=theorem]{lemma}
\declaretheorem[name=Proposition, sibling=theorem]{proposition}

\declaretheorem[name=Corollary, sibling=theorem]{corollary}

\declaretheorem[name=Claim, sibling=theorem]{claim}

\declaretheorem[name=Question, style=remark, sibling=theorem]{question}
\def\cqedsymbol{\ifmmode$\lrcorner$\else{\unskip\nobreak\hfil
\penalty50\hskip1em\null\nobreak\hfil$\lrcorner$
\parfillskip=0pt\finalhyphendemerits=0\endgraf}\fi} 

\newcommand{\tri}{\begin{tikzpicture}[scale=0.2]
    \node[draw, circle] (A) at (30:2) {};
    \node[draw, circle] (B) at (150:2) {};
    \node[draw, circle] (C) at (-90:2) {};
    \draw[thick] (0,0) circle (3.7cm);
    \draw[->] (A) to (B);
    \draw[->] (B) to (C);
    \draw[->] (C) to (A);
  \end{tikzpicture}}

\usepackage[linesnumbered,ruled,onelanguage,vlined]{algorithm2e}

\let\le\leqslant
\let\ge\geqslant

\definecolor{ForestGreen}{RGB}{0,205,102}

\title{The smallest 5-chromatic tournament
\footnote{The first author was supported by the European Research Council (ERC) under the European Union's Horizon 2020 research and innovation programme Grant Agreement 714704. \\
The third and fourth authors were supported by the MUNI Award in Science and Humanities of the Grant Agency of Masaryk University. \\
The third author was also supported by the project GA20-09525S of the Czech Science Foundation. \\
The second and fourth authors were supported by ANR project GrR (ANR-18-CE40-0032).}}

\author[1,4]{Thomas Bellitto}
\author[2]{Nicolas Bousquet}
\author[3,5]{Adam Kabela}
\author[2,5]{Théo Pierron}
\affil[1]{Sorbonne Université, CNRS, LIP6, F-75005 Paris, France, \texttt{thomas.bellitto@lip6.fr}}
\affil[2]{Univ. Lyon, Université Lyon 1, LIRIS UMR CNRS 5205, F-69621, Lyon, France, \texttt{firstname.lastname@univ-lyon1.fr}}
\affil[3]{Faculty of Applied Sciences, University of West Bohemia, Pilsen, Czech Republic, \texttt{kabela@kma.zcu.cz}}
\affil[4]{Previous affiliation: Faculty of Mathematics, Informatics and Mechanics, University of Warsaw, Poland}
\affil[5]{Previous affiliation: Faculty of Informatics, Masaryk University, Brno, Czech Republic}
\date{}

\begin{document}

\maketitle

\begin{abstract}
A coloring of a digraph is a partition of its vertex set such that each class induces a digraph with no directed cycles.
A digraph is $k$-chromatic if $k$ is the minimum number of classes in such partition, 
and a digraph is oriented if there is at most one arc between each pair of vertices.
Clearly, the smallest $k$-chromatic digraph is the complete digraph on $k$ vertices,
but determining the order of the smallest $k$-chromatic oriented graphs is a challenging problem.
It is known that the smallest $2$-, $3$- and $4$-chromatic oriented graphs have $3$, $7$ and $11$ vertices, respectively.
In 1994, Neumann-Lara conjectured that a smallest $5$-chromatic oriented graph has $17$ vertices.
We solve this conjecture and show that the correct order is $19$.

{\bf MS Classification: 05C20}
\end{abstract}

\section{Introduction}\label{sec:intro}

Finding proper colorings of graphs lies among the most studied problems in graph theory. The goal consists in coloring vertices so that adjacent ones receive distinct colors.
In \cite{def_dicoloring}, Neumann-Lara introduced a generalization of this problem to digraphs. A digraph consists of a vertex set $V$ plus a set of (ordered) pairs of vertices called \emph{arcs}.
Graphs can be seen as a special case of digraphs where, for every arc $uv$, there also exists an arc $vu$ (such digraphs are called symmetric and these pairs of arcs are called \emph{digons}). 

Neumann-Lara defines a proper coloring of a digraph as a partition of the vertex set into acyclic sets (i.e., subsets of vertices which do not contain any oriented cycle). Note that when all the arcs come in digons, this notion indeed reduces to the usual graph coloring definition of unoriented graphs.
The smallest number of colors required to color properly a digraph $D$ is called the \emph{chromatic number of $D$} and will be denoted by $\chi(D)$ in the rest of the paper\footnote{Note that it is sometimes denoted $\overrightarrow{\chi}(D)$, especially when confusion is possible with chromatic number of unoriented graphs. Since we consider only digraphs we prefer keeping the notation as simple as possible.}.
While there exist other generalizations of coloring to digraphs for instance based on graph homomorphisms, see e.g.~\cite{courcelle}, Neumann-Lara's is the most classical one and the one that received ever-growing attention since its introduction.

An \emph{oriented graph} is a digraph that does not contain any digon. Understanding the behavior and structure of graphs of small order and large chromatic number becomes much harder on oriented graphs. Indeed, for any integer $k$, the smallest digraph of chromatic number $k$ is the complete graph on $k$ vertices (i.e. there is a digon between every pair of vertices). On the contrary, determining the order of the smallest oriented graph of chromatic number $k$ was already raised by Neumann-Lara in 1982 in~\cite{NL-tournaments}. 
The goal of this paper is to tackle that problem for $k=5$.

Observe that adding arcs to an oriented graph cannot decrease its chromatic number. Therefore, if there exists an oriented graph of chromatic number $k$ there exists a tournament of chromatic number $k$; a \emph{tournament} being an orientation of the undirected complete graph. Then, the quest of the smallest oriented graph of chromatic number $k$ can be restricted to tournaments. 
A tournament (or by abuse of notations a set of vertices) is \emph{transitive} if for every pair of arcs $uv$ and $vw$, the arc between $uw$ also exists. One can easily observe that a tournament can be colored with $k$ colors if and only if it can be partitioned in $k$ transitive subtournaments. Thus, Neumann-Lara's question can be rephrased as follows:

\begin{question}\label{question}
For every $k$, what is the smallest value of $n_k$ for which there exists a tournament $D$ on $n_k$ vertices which cannot be partitioned into $(k-1)$ transitive subtournaments?
\end{question}

This formulation connects this problem to questions raised by Erd{\H o}s and Moser about 20 years before Neumann-Lara's question~\cite{ErdosMoser}.
The question was asymptotically solved since the maximum chromatic number of a tournament on $n$ vertices is $\Theta(\frac n {\log n})$~\cite{ErdosNL1,ErdosMoser}. However, the question of determining the exact values even for small values of $k$ is still widely open. 

The smallest tournament of chromatic number $2$ (\textit{i.e.} non-transitive) is the directed cycle of length $3$. The constructions for $k=4,5$ rely on the so-called Paley tournaments. For every prime integer $n$ of the form $4k+3$, the \emph{Paley tournaments} on $n$ vertices $\mathrm{Pal}_n$ is the tournament whose vertex set is $\{0,...,n-1\}$ and containing the arc $ij$ if and only if $i-j$ is a square modulo $n$. 
In \cite{NL-tournaments}, Neumann-Lara proved that the smallest tournament of chromatic number $3$ has order $7$ and that there exist four such tournaments, including $\mathrm{Pal}_7$. He also proved that the smallest tournament of chromatic number $4$ has order $11$, is unique and is actually $\mathrm{Pal}_{11}$. 
In the conclusion of \cite{NL-tournaments}, Neumann-Lara discussed the possible order of the smallest $5$-chromatic tournament. He claimed to know that the answer is between $17$ and $19$ and conjectured that it is $17$. Note that the next "natural" candidate, namely the Paley tournament $\mathrm{Pal}_{19}$ is actually $4$-colorable. Neumann-Lara actually published his construction of a $5$-chromatic tournament on $19$ vertices six years later in~\cite{neumann2000dichromatic}.

Even if the question received a considerable attention and was mentioned often as an open problem in the literature in the last $30$ years (see~\cite{dicritical,Kostochka} for recent examples),
determining the exact value of $n_5$ is still open today.

The goal of this paper is to answer this question and provides a definitive answer to Neumann-Lara's question for $k=5$. Namely,

\begin{theorem}
The smallest order of a $5$-chromatic tournament is $19$. 
\end{theorem}

After presenting some tools in Sections~\ref{sec:tools} and~\ref{sec:12}, we disprove Neumann-Lara's conjecture in Section~\ref{sec:17} by showing that every tournament on $17$ vertices is $4$-colorable. The proof relies on a surprising intermediate result (Theorem~\ref{thm:contains P11}) of independent interest proved by a computer analysis. Namely all the $4$-chromatic tournaments on $12$ vertices contain $\mathrm{Pal}_{11}$ as a subtournament. We derive from it that all the tournaments on $17$ vertices have chromatic number $4$ with a short and human-readable proof. We leave as an open problem a human-readable proof that all the $4$-tournaments on $12$ vertices contain $\mathrm{Pal}_{11}$ as a subtournament. So $n_5 \ge 18$.

We then exhibit an example of a $5$-chromatic tournament on $19$ vertices in Section~\ref{sec:19}, which ensures that $n_5 \le 19$.

We finally present in Section~\ref{sec:18} a computer-assisted proof showing that all the tournaments on $18$ vertices are $4$-colorable, which ensures that $n_5=19$ and settles the case $k=5$.

Note that the number of non-isomorphic tournaments on $17$, $18$ and $19$ vertices have respectively $27$, $31$ and $35$ digits \cite{oeis}, generating them up to isomorphism is already a very challenging task and the problem of $5$-colorability that we need to solve on each of them is NP-complete. Therefore, it is definitely out of reach to solve the problem by bruteforce. Instead, we use the approach summarized in the following sketch. We observe that any $5$-chromatic $18$-vertex tournament must contain two or three (vertex-)disjoint copies of $TT_5$, the transitive tournament of order $5$. In the latter case, we can thus decompose its vertex set as $A_1,A_2,A_3,B$ such that each $A_i$ induces $TT_5$ and $B$ induces a directed triangle. We may observe that $A_i\cup B$ induces a $3$-chromatic tournament on $8$ vertices and that $A_i\cup A_j\cup B$ induces a $4$-chromatic tournament on $13$ vertices. We proceed as follows: 
\begin{enumerate}
    \item We generate the so-called \emph{$8$-completions}, that are the non-isomorphic ways to orient the arcs between a $TT_5$ and a directed triangle so that the resulting tournament is $3$-chromatic. To this end, we use a branching algorithm involving a trimming operation when we detect 
    that the branch will not generate any $3$-chromatic tournament.
There are $256$ such $8$-completions.
    \item For each pair of $8$-completions, we identify their distinguished directed triangle and we generate the possible orientations of the arcs between their respective $TT_5$ so that the result is a $4$-chromatic $13$-vertex tournament (with two distinguished copies of $TT_5$). These are called \emph{$13$-completions}. For each pair of $8$-completions, the maximum amount of $13$-completions is $2072$. However only $4508$ pairs have at least one $13$-completion (with an average of 47.6 completions for each pair), and a quarter of them has precisely one $13$-completion.
    \item We then consider the triples of $8$-completions $(C_1,C_2,C_3)$ where for each $1\leqslant i<j\leqslant 3$, $(C_i,C_j)$ lies among these $4508$ pairs. We generate all the $18$-vertex tournaments obtained by identifying their distinguished directed triangle, and adding the arcs of each $13$-completion between $C_i$ and $C_j$ for $i,j\in\{1,2,3\}$. We finally check whether one of these tournaments is $5$-chromatic. 
\end{enumerate}

When our candidate has exactly two disjoint copies of $TT_5$, note that the remaining $8$-vertex tournament $X_8$ must be $3$-chromatic and without $TT_5$, hence lies among a list of only 94 elements. For each such tournament $X_8$, we proceed as follows:
\begin{enumerate}
    \item We re-use our branching algorithm to generate all orientations of arcs between $X_8$ and $TT_5$ that yield a $4$-chromatic $13$-vertex tournament (and adapt the trimming to cut the branch when we detect two disjoint $TT_5$). 
    \item For each pair of such orientations, we generate a $13$-vertex oriented graph by identifying their common $8$-vertex tournament.
    \item We discard \emph{incompatible pairs}, that are pairs with a $4$-coloring where colors $1,2$ are used only on one $TT_5$ and $3,4$ on the other. Indeed, for these pairs, every orientation of the remaining arcs will stay $4$-colorable. 
    \item For each remaining pair, we try all possible orientations of the remaining $25$ arcs and check whether the resulting tournament is $5$-chromatic (re-using our branching algorithm). 
\end{enumerate}

For $17$-vertex tournaments, we could follow roughly the same approach as in the two $TT_5$ case (since we can show each candidate must contain two copies, but cannot contain three of them). However, we can show by hand that all pairs are incompatible in Step $3$, which directly concludes without using Step $4$. To this end, we do not need to consider the full output of Step 2, but only the intermediate result stated in Section~\ref{sec:12}.

\paragraph{The case $k=6$.} It was already known to Neumann-Lara that $n_6\leqslant 26$ since there is a $TT_6$-free tournament on $26$-vertices, which is thus not $5$-colorable. Our result actually implies that $n_6\geqslant 24$ since every tournament on $23$ vertices must contain $TT_5$, and the remaining vertices induce a $4$-colorable tournament. Besides these easy observations, nothing seems to be known about the exact value of $n_6$.

\paragraph{Related work.}

In \cite{dicritical}, Bang-Jensen \textit{et al.} establish some structural results about \emph{$k$-critical digraphs}, \textit{i.e.} digraphs with chromatic number $k$ that are minimal by inclusion. 
The average degree of such digraphs was also source of attention in recent years. In \cite{Hoshino,Kostochka}, the authors provide some bounds on the smallest possible value of this parameter among all $k$-critical digraphs on $n$ vertices. Note that the question is easily answered without the dependency in $n$ since each vertex of a $k$-critical digraph needs to have in- and out-degree at least $k-1$ and this value is reached by complete digraphs on $k$ vertices. But here again, the question becomes much more difficult when digons are forbidden: the smallest average degree of oriented $k$-critical graphs is still open even for $k=3$~\cite{4.67}.

These works are also reminiscent of numerous works in the undirected case that look for the smallest graph of chromatic number $k$ that does not contain any complete subgraph of order $c$. 
The problem has been especially well-studied for triangle-free graphs (the case $c=3$), since this is the smallest value of $c$ that makes the problem non-trivial. 
In \cite{Chvatal}, Chvátal proved that the smallest triangle-free $k$-chromatic graph has order $11$ for $k=4$ and Jensen and Royle proved in \cite{JensenRoyle} through a computer search that it has order $22$ for $k=5$. 
The question is still open for $k=6$ where Goedgebeur proved in \cite{Goedgebeur} that it is between 32 and 40. 
For digraphs, forbidding cliques of size $c=2$ corresponds to considering oriented graphs, and actually yields again Neumann-Lara's question we study in this paper.

\section{Tools}
\label{sec:tools}

In this section, we introduce structural results that are used throughout the paper.
The first of these results answers the question on the smallest $3$-chromatic tournaments. 

\begin{theorem}[\cite{NL-tournaments}]
\label{thm:T7 3chrom}
Every tournament on $6$ vertices is $2$-colorable.
Moreover, there are exactly four $3$-chromatic tournaments on $7$ vertices;
namely, the tournament $\mathrm{Pal}_7,W,W_0,W_1$ depicted in Figure~\ref{fig:T7}.
\end{theorem}

\begin{figure}[!ht]
\centering
\begin{tikzpicture}[every node/.style={draw, circle,inner sep=1pt},thick,scale=1.25]
\node (0) at (0,1) {0};
\node (4) at (-.75,0.5) {4};
\node (5) at (-1.61,0) {5};
\node (6) at (-.75,-0.5) {6};
\node (1) at (.75,0.5) {1};
\node (2) at (1.61,0) {2};
\node (3) at (.75,-0.5) {3};
\draw[-{stealth}] (0) -> (1);
\draw[-{stealth},bend left] (0) to (2);
\draw[-{stealth}] (0) -> (3);
\draw[-{stealth}] (4) -> (0);
\draw[-{stealth},bend left] (5) to (0);
\draw[-{stealth}] (6) -> (0);
\draw[-{stealth}] (1) -> (2);
\draw[-{stealth}] (2) -> (3);
\draw[-{stealth}] (3) -> (1);
\draw[-{stealth}] (4) -> (5);
\draw[-{stealth}] (5) -> (6);
\draw[-{stealth}] (6) -> (4);
\draw[-{stealth}] (1) -> (5);
\draw[-{stealth}] (1) -> (6);
\draw[-{stealth}] (4) -> (1);
\draw[-{stealth}] (2) -> (5);
\draw[-{stealth}] (2) -> (4);
\draw[-{stealth}] (6) -> (2);
\draw[-{stealth}] (3) -> (4);
\draw[-{stealth}] (3) -> (6);
\draw[-{stealth}] (5) -> (3);
\node[color=white] at (-2.5,0) {\textcolor{black}{$\mathrm{Pal}_7$}};
\tikzset{xshift=5cm}
\node (0) at (0,1) {0};
\node (4) at (-.75,0.5) {4};
\node (5) at (-1.61,0) {5};
\node (6) at (-.75,-0.5) {6};
\node (1) at (.75,0.5) {1};
\node (2) at (1.61,0) {2};
\node (3) at (.75,-0.5) {3};
\draw[-{stealth}] (0) -> (1);
\draw[-{stealth},bend left] (0) to (2);
\draw[-{stealth}] (0) -> (3);
\draw[-{stealth}] (4) -> (0);
\draw[-{stealth},bend left] (5) to (0);
\draw[-{stealth}] (6) -> (0);
\draw[-{stealth}] (1) -> (2);
\draw[-{stealth}] (2) -> (3);
\draw[-{stealth}] (3) -> (1);
\draw[-{stealth}] (4) -> (5);
\draw[-{stealth}] (5) -> (6);
\draw[-{stealth}] (6) -> (4);
\draw[-{stealth}] (1) -> (5);
\draw[-{stealth}] (1) -> (6);
\draw[-{stealth}] (1) -> (4);
\draw[-{stealth}] (2) -> (5);
\draw[-{stealth}] (2) -> (4);
\draw[-{stealth}] (2) -> (6);
\draw[-{stealth}] (3) -> (4);
\draw[-{stealth}] (3) -> (6);
\draw[-{stealth}] (3) -> (5);
\node[color=white] at (2.5,0) {\textcolor{black}{$W$}};
\tikzset{yshift=-2cm,xshift=-5cm}
\node (0) at (0,1) {0};
\node (4) at (-.75,0.5) {4};
\node (5) at (-1.61,0) {5};
\node (6) at (-.75,-0.5) {6};
\node (1) at (.75,0.5) {1};
\node (2) at (1.61,0) {2};
\node (3) at (.75,-0.5) {3};
\draw[-{stealth}] (0) -> (1);
\draw[-{stealth},bend left] (0) to (2);
\draw[-{stealth}] (0) -> (3);
\draw[-{stealth}] (4) -> (0);
\draw[-{stealth},bend left] (5) to (0);
\draw[-{stealth}] (6) -> (0);
\draw[-{stealth}] (1) -> (2);
\draw[-{stealth}] (2) -> (3);
\draw[-{stealth}] (3) -> (1);
\draw[-{stealth}] (4) -> (5);
\draw[-{stealth}] (5) -> (6);
\draw[-{stealth}] (6) -> (4);
\draw[-{stealth}] (1) -> (5);
\draw[-{stealth}] (1) -> (6);
\draw[-{stealth}] (4) -> (1);
\draw[-{stealth}] (2) -> (6);
\draw[-{stealth}] (2) -> (4);
\draw[-{stealth}] (5) -> (2);
\draw[-{stealth}] (3) -> (4);
\draw[-{stealth}] (3) -> (5);
\draw[-{stealth}] (6) -> (3);
\node[color=white] at (-2.5,0) {\textcolor{black}{$W_0$}};
\tikzset{xshift=5cm}
\node (0) at (0,1) {0};
\node (4) at (-.75,0.5) {4};
\node (5) at (-1.61,0) {5};
\node (6) at (-.75,-0.5) {6};
\node (1) at (.75,0.5) {1};
\node (2) at (1.61,0) {2};
\node (3) at (.75,-0.5) {3};
\draw[-{stealth}] (0) -> (1);
\draw[-{stealth},bend left] (0) to (2);
\draw[-{stealth}] (0) -> (3);
\draw[-{stealth}] (4) -> (0);
\draw[-{stealth},bend left] (5) to (0);
\draw[-{stealth}] (6) -> (0);
\draw[-{stealth}] (1) -> (2);
\draw[-{stealth}] (2) -> (3);
\draw[-{stealth}] (3) -> (1);
\draw[-{stealth}] (4) -> (5);
\draw[-{stealth}] (5) -> (6);
\draw[-{stealth}] (6) -> (4);
\draw[-{stealth}] (1) -> (5);
\draw[-{stealth}] (1) -> (6);
\draw[-{stealth}] (1) -> (4);
\draw[-{stealth}] (2) -> (6);
\draw[-{stealth}] (2) -> (4);
\draw[-{stealth}] (5) -> (2);
\draw[-{stealth}] (3) -> (4);
\draw[-{stealth}] (3) -> (5);
\draw[-{stealth}] (6) -> (3);
\node[color=white] at (2.5,0) {\textcolor{black}{$W_1$}};
\end{tikzpicture}
\caption{The $3$-chromatic tournaments on $7$ vertices.}
\label{fig:T7}
\end{figure}
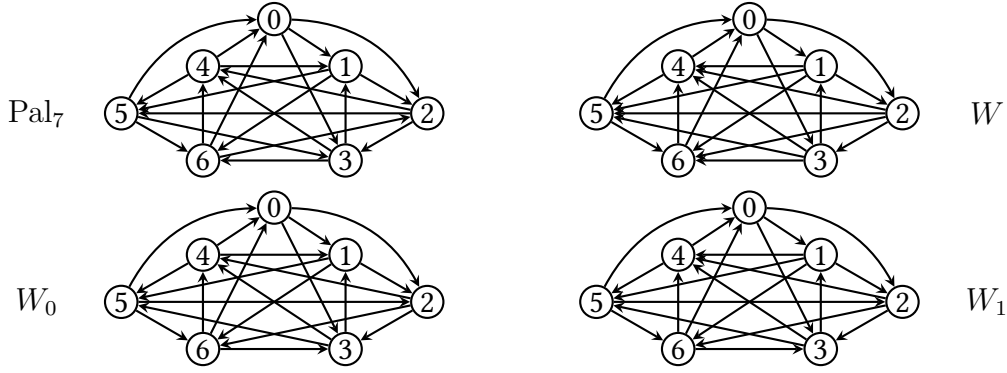

We then investigate tournaments which contain no transitive subtournaments of prescribed order.
We let $TT_k$ denote the transitive tournament on $k$ vertices, and we say that an oriented graph is \emph{$TT_k$-free} if it does not contain $TT_k$ as a subgraph. A simple inductive argument (using that the in-neighborhood or the out-neighborhood of each vertex contains at last half of the other vertices of a tournament) yields the following.
\begin{lemma}[\cite{stearns1959voting}]
\label{lem:sub}
Every tournament on $2^{k-1}$ vertices contains $TT_k$.
\end{lemma}

The bound of Lemma~\ref{lem:sub} is not tight, and determining the order of smallest $TT_k$-free tournaments is an open question. A precise answer is known only for $k\leqslant 6$~\cite{sanchez1998tournaments}, and we need the case $k=5$.

\begin{theorem}[\cite{REID1970225}]
\label{thm:14}
Every tournament on $14$ vertices contains $TT_5$.
\end{theorem}

Interestingly, there is precisely one $TT_5$-free tournament on $12$ vertices,
and precisely one $TT_5$-free tournament on $13$ vertices
as shown by the following. 

\begin{theorem}[\cite{sanchez1998tournaments}]
\label{thm:unique T12}
There is a unique $TT_5$-free tournament on $12$ vertices, and its chromatic number is $3$.
\end{theorem}

\begin{theorem}[\cite{REID1970225}]
\label{thm:X13}
There is a unique $TT_5$-free tournament on $13$ vertices,
and it can be represented so that the vertices are integers $0, \dots, 12$
and $ij$ is an arc if and only if $j-i\in\{1,2,3,5,6,9\}$ modulo $13$.
\end{theorem}

We let $X_{13}$ be the unique $TT_5$-free tournament on $13$ vertices,
and we conclude this section with two propositions on the properties of $X_{13}$. It is well-known that this tournament is \emph{vertex-transitive}, which means that its automorphism group acts transitively on its vertices. The following stronger result actually holds.

\begin{proposition}[\cite{REID1970225}]
\label{clm:struct_X13}
The tournament $X_{13}$ is vertex-transitive and for every arc $ij$, there exists an automorphism of $X_{13}$ mapping $ij$ to either $01$ or $02$.
\end{proposition}

Using Proposition~\ref{clm:struct_X13}, we determine the structure of the copies of $TT_4$ in $X_{13}$ as follows. 

\begin{proposition}
\label{claim:main6.1}
Let $A$ be a set of vertices inducing $TT_4$ in $X_{13}$ whose vertices of highest out-degree are either $\{0,1\}$ or $\{0,2\}$. Then $A$ is either $\{0,1,2,3\}$, $\{0,1,3,6\}$, $\{0,1,6,2\}$ or $\{0,2,3,5\}$. Moreover, the four possible tournaments obtained by removing $A$ from $X_{13}$ are pairwise non-isomorphic, and each of them has no non-trivial automorphism. 
\end{proposition}

\begin{proof}
Let $a,b,c$ and $d$ be four vertices in transitive order in $X_{13}$, that is, $ab, ac, ad$ and $bc, bd$ and $cd$ are arcs of $X_{13}$.
By hypothesis, we have $ab=01$ or $ab=02$. Now we use that $c$ and $d$ are out-neighbors of both $a$ and $b$.
If the arc $ab$ is $01$, then we can only complete $01$ into a $TT_4$ by choosing $cd$ as $23,36$ or $62$.
If $ab$ is $02$, then the only way to complete $02$ into a $TT_4$ is with $cd$ chosen as $35$. This concludes the first part of the statement.

We let $A_1 = \{0,1,2,3\}$, $A_2 = \{0,1,3,6\}$, $A_3 = \{0,1,6,2\}$ and $A_4 = \{0,2,3,5\}$,
and let $T_i = X_{13} \setminus A_i$ for every $i$ of $\{0,1,2,3\}$.

In order to show that the tournaments $T_1, T_2, T_3$ and $T_4$ are pairwise non-isomorphic,
we consider the subtournament $T_i'$ of $T_i$ induced by the set of all vertices whose in-degree in $T_i$ is $4$
(see tournaments $T_1', \dots, T_4'$ depicted in Figure~\ref{fig:nonisom}). 
We note that $T_1', \dots, T_4'$ are pairwise non-isomorphic, and thus $T_1, \dots, T_4$ are pairwise non-isomorphic.

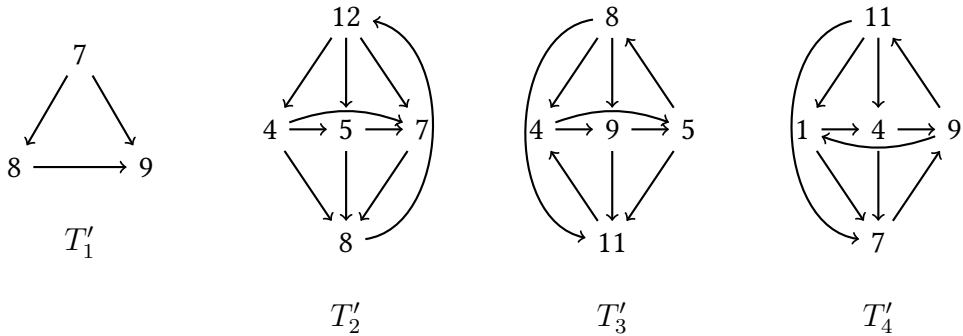
\begin{figure}[!ht]
\centering
\begin{tikzpicture}[thick,v/.style={minimum size = 10pt, circle}]
\node (7) at (90:1) {7};
\node (8) at (210:1) {8};
\node (9) at (330:1) {9};
\draw[->] (7) to (8);
\draw[->] (7) to (9);
\draw[->] (8) to (9);
\node at (0,-1.5) {$T_1'$};
\tikzset{xshift=3.5cm}
\node (12) at (0,1.5) {12};
\node (8) at (0,-1.5) {8};
\node (4) at (-1,0) {4};
\node (5) at (0,0) {5};
\node (7) at (1,0) {7};
\draw[->] (12) to (4);
\draw[->] (12) to (5);
\draw[->] (12) to (7);
\draw[->] (4) to (8);
\draw[->] (5) to (8);
\draw[->] (7) to (8);
\draw[->] (4) to (5);
\draw[->] (5) to (7);
\draw[->,bend left=20] (4) to (7);
\draw[->,bend right=80] (8) to (12);
\node at (0,-2.5) {$T_2'$};
\tikzset{xshift=3.5cm}
\node (8) at (0,1.5) {8};
\node (11) at (0,-1.5) {11};
\node (4) at (-1,0) {4};
\node (9) at (0,0) {9};
\node (5) at (1,0) {5};
\draw[->] (8) to (4);
\draw[->] (5) to (8);
\draw[->] (8) to (9);
\draw[->] (11) to (4);
\draw[->] (5) to (11);
\draw[->] (9) to (11);
\draw[->] (4) to (9);
\draw[->] (9) to (5);
\draw[->,bend left=20] (4) to (5);
\draw[->,bend right=80] (8) to (11);
\node at (0,-2.5) {$T_3'$};
\tikzset{xshift=3.5cm}
\node (8) at (0,1.5) {11};
\node (11) at (0,-1.5) {7};
\node (4) at (-1,0) {1};
\node (9) at (0,0) {4};
\node (5) at (1,0) {9};
\draw[->] (8) to (4);
\draw[->] (5) to (8);
\draw[->] (8) to (9);
\draw[->] (11) to (5);
\draw[->] (4) to (11);
\draw[->] (9) to (11);
\draw[->] (4) to (9);
\draw[->] (9) to (5);
\draw[->,bend left=20] (5) to (4);
\draw[->,bend right=80] (8) to (11);
\node at (0,-2.5) {$T_4'$};
\end{tikzpicture}
\caption{Tournaments $T_1', \dots, T_4'$.}
\label{fig:nonisom}
\end{figure}

Finally, we show that each of $T_1, \dots, T_4$ has no non-trivial automorphism.
For $T_1$,
we observe that vertices $4,5,6$ have in-degree $3$, vertices $7,8,9$ have in-degree $4$, and 
vertices $10,11,12$ have in-degree $5$ in $T_1$.
Furthermore, each of the sets $\{4,5,6\}$, $\{7,8,9\}$, and $\{10,11,12\}$ induces a $TT_3$.
Since automorphisms preserve in-degrees and $TT_3$ has no non-trivial automorphism,
we conclude that each vertex of $T_1$ has to be mapped to itself in every automorphism of $T_1$.
Thus, $T_1$ has no non-trivial automorphism.

For every $i$ of $\{2,3,4\}$,
we note that $T_i$ has precisely two vertices of in-degree $3$ and precisely two vertices of in-degree $5$ in $T_i$.
In particular, each of these vertices has to be mapped to itself in every automorphism of $T_i$.
We recall that the remaining vertices induce $T_i'$, and it remains to show that $T_i'$ has no non-trivial automorphism. 

For every $i$ of $\{2,3,4\}$,
we note that $T_i'$ contains precisely one vertex of in-degree $1$ and precisely one vertex of in-degree $3$ in $T_i'$,
and each of these vertices has to be mapped to itself in every automorphism of $T_i'$.
For each of $T_2'$ and $T_3'$,
we note that the set of all vertices of in-degree $2$ induces a $TT_3$.
It follows that each of $T_2'$ and $T_3'$ has no non-trivial automorphism.
For $T_4'$, we observe that vertex $9$ has to be mapped to itself in every automorphism of $T_4'$
(since there is an arc from $9$ to the unique vertex of in-degree $1$ in $T_4'$).
The desired conclusion for $T_4'$ follows.
\end{proof}

\section{The $4$-chromatic tournaments on $12$ vertices}
\label{sec:12}
Our disproof of Neumann-Lara's conjecture heavily relies on the following result, that is interesting by itself and has already been useful for other projects. 

\begin{theorem}
\label{thm:contains P11}
Every $4$-chromatic tournament on $12$ vertices contains $\mathrm{Pal}_{11}$.
\end{theorem}

This result has already been proven of interest; the first author used it for another project with other co-authors in~\cite{4.67}. The authors show that for every $k\geqslant 2$, there exist $k$-critical oriented graphs of any possible order larger than some threshold $p_k$. They then used Theorem~\ref{thm:contains P11} to prove that there is no $4$-critical oriented graphs on $12$ vertices (while $\mathrm{Pal}_{11}$ is one on $11$ vertices). In particular, this implies that $p_k$ is not necessarily the order $n_k$ of a smallest $k$-critical oriented graph (which is actually true for $k=2$ and $3$).

Our proof of Theorem~\ref{thm:contains P11} relies on a computer program, that basically went through an extensive case analysis, that would be too long to do by hand. In this section, we introduce the ideas behind the program. These ideas will then be reused and developed to prove Theorem~\ref{thm:18} in Section~\ref{sec:18}. All our programs can be found at~\url{https://github.com/tpierron/5chromatictournaments/}. 

Let $T$ be a $4$-chromatic tournament on $12$ vertices.
By Theorem~\ref{thm:unique T12}, $T$ contains a $TT_5$.
Since $T$ is $4$-chromatic, the remaining $7$ vertices induce a $3$-chromatic tournament.
Using Theorem~\ref{thm:T7 3chrom},
it follows that the vertices of $T$ can be partitioned into a copy of $TT_5$
and a copy of a tournament $X$ among $\{\mathrm{Pal}_7,W,W_0,W_1\}$.
In particular, we say that $T$ is a \emph{gluing} of $TT_5$ and $X$.

A naive way to prove Theorem~\ref{thm:contains P11} is then to try all the $4\times 2^{35}$ possible gluings, keep the $4$-chromatic ones and check whether they all contain $\mathrm{Pal}_{11}$. While this is almost doable, we explain here how to make this process faster, so that Theorem~\ref{thm:contains P11} can be checked in a matter of hours on a standard computer. 

Fix a tournament $X\in\{\mathrm{Pal}_7,W,W_0,W_1\}$. Instead of generating the $2^{35}$ gluings of $X$ with $TT_5$ directly, and then filter out the $3$-colorable ones, we generate them using a branching algorithm (see Algorithm~\ref{algo:dsmash}) in such a way that we will be able to cut branches.

\begin{algorithm}[H]
\KwIn{An oriented graph $T$.}
\KwOut{All arc-extensions of $T$ to $4$-chromatic tournaments on $V(T)$.}
\BlankLine
\If{$V(T)$ cannot be partitioned into 3 transitive tournaments}{
\eIf{$T$ is a tournament}
{\Return{$[T]$}}
{Choose two non-adjacent vertices $a,b$ in $T$.

\Return{\texttt{completions}$(T+ab)$ $+$ \texttt{completions}$(T+ba)$}
}
}
\caption{\texttt{completions}$(T)$}
\label{algo:dsmash}
\end{algorithm}

We start from disjoint union of $X$ and $TT_5$, and apply \texttt{completions}.
At each step, we choose a pair of non-adjacent vertices and add an arc joining them
which gives two branches of the computation (one branch for each possible direction of the new arc).
The main observation is that, if at some point $V(T)$ can be partitioned into three transitive tournaments,
then we can immediately cut the branch since all the tournaments we could obtain from this point onwards will be $3$-colorable. 

In order to prove Theorem~\ref{thm:contains P11}, we just run \texttt{completions} four times (once for each choice of $X$), and then check for a $\mathrm{Pal}_{11}$ in each of the resulting tournaments. One can check that when $X\neq W_1$, this step is not needed since the output is always empty, i.e. all the gluings are $3$-colorable. This yields the following by-product of our proof, which can actually be deduced from Theorem~\ref{thm:contains P11} (while not being necessary for proving it).

\begin{corollary}
\label{cor:gluing}
Every $4$-chromatic tournament on $12$ vertices is a gluing of $W_1$ and $TT_5$.
\end{corollary}

\begin{proof}
Let $T$ be a $4$-chromatic tournament on $12$ vertices. By Theorem~\ref{thm:contains P11}, there is a vertex $x$ such that $T-x$ is $\mathrm{Pal}_{11}$. Moreover, by Lemma~\ref{lem:sub}, $T$ also contains a set $S$ of 5 vertices inducing $TT_5$. Since $\mathrm{Pal}_{11}$ is $TT_5$-free, $S$ must contain $x$.

Now, the four remaining vertices of $S\setminus\{x\}$ induce $TT_4$ in $\mathrm{Pal}_{11}$. Observe that $\mathrm{Pal}_{11}$ is arc-transitive, hence up to renaming vertices, we may assume that $0$ and $1$ are the first and second vertices in the $TT_4$. One can then easily check that $S=\{x,0,1,4,5\}$. Now observe that the remaining vertices of $\mathrm{Pal}_{11}$ induce $W_1$. 
\end{proof}

We conclude this section by outlining the implementation of the $3$-colorability test.
The full implementation can be found in~\href{https://github.com/tpierron/5chromatictournaments/blob/main/section3.ml}{section3.ml}. 

If $T$ is $k$-colorable,
one can choose a $k$-coloring such that the size of the first color class is maximized.    
Therefore, with the list $L$ of sets of vertices inducing maximal transitive subtournaments of $T$
we can test $k$-colorability recurrently as follows: for each $V$ from $L$,
we check if $T - V$ is $(k-1)$-colorable.
The list of maximal transitive subtournaments of $T - V$
can be obtained by removing the vertices of $V$ from the elements of $L$. 
    
In particular, when running \texttt{completions}, we do not recompute the list $L$ from scratch at each call. Instead, we just update it when adding an arc.

\section{Disproving Neumann-Lara's conjecture}\label{sec:17}

\begin{theorem}
\label{thm:17-4}
Every $17$-vertex tournament is $4$-colorable.
\end{theorem}

This section is devoted to the proof of Theorem~\ref{thm:17-4}. By contradiction, we consider a $17$-vertex tournament $T_{17}$ which is not $4$-colorable. We show that, due to this assumption, $T_{17}$ has a very rigid structure, which allows us reach a contradiction by constructing a $4$-coloring of $T_{17}$. The first structural result is summarized in the following lemma.

\begin{lemma}
\label{lem:2TT5in17}
 The vertices of $T_{17}$ can be partitioned in three sets $A_1,A_2,B$ such that $A_1$ and $A_2$ both induce  $TT_5$ and $B$ induces $W_1$.
\end{lemma}

\begin{proof}
By Lemma~\ref{lem:sub}, every tournament on $17$ vertices contains a set $A_1$ of five vertices inducing a transitive tournament. Removing $A_1$ from $T_{17}$ gives a tournament $T_{12}$ on $12$ vertices, which is not $3$-colorable (otherwise $T_{17}$ would be $4$-colorable). The result now follows by applying Corollary~\ref{cor:gluing} to $T_{12}$.
\end{proof}

The contradiction then follows directly from the next lemma.
\begin{lemma}
\label{lem:2col}
One can split $B$ as $B_1\cup B_2$ such that each of the subtournaments of $T_{17}$ induced by $A_1\cup B_1$ and $A_2\cup B_2$ is $2$-colorable.
\end{lemma}

The rest of the proof is devoted to prove Lemma~\ref{lem:2col}.
To prove this lemma, we consider the tournament $T_{12}$ induced by $A_1\cup B$, and we identify some subsets $B_1$ of $B$ such that $A_1\cup B_1$ induces a $2$-colorable tournament. Up to renaming, we can assume that $B=[0,6]$ (with the labeling depicted in Figure~\ref{fig:T7}). Let us first state  four claims whose proofs are postponed to the end of this section. For readability, we write $\chi(X)$ to denote the chromatic number of the oriented graph induced in $T_{12}$ by a set $X$ of vertices.

\begin{claim}
\label{claim:1}
$\chi(A_1\cup\{0,1,4\})=2$.
\end{claim}

\begin{claim}
\label{claim:2}
$\chi(A_1\cup\{0,1,2,3\})=2$ or $\chi(A_1\cup\{0,4,5,6\})=2$.
\end{claim}

\begin{claim}
\label{claim:3}
If $\chi(A_1\cup\{4,5,6\})>2$ and $\chi(A_1\cup\{2,3,5,6\})>2$, then $\chi(A_1\cup\{0,2,4\})=\chi(A_1\cup\{1,3,5,6\})=2$.
\end{claim}

\begin{claim}
\label{claim:4}
If $\chi(A_1\cup\{1,2,3\})>2$ and $\chi(A_1\cup\{2,3,5,6\})>2$, then $\chi(A_1\cup\{0,1,6\})=\chi(A_1\cup\{2,3,4,5\})=2$.
\end{claim}
Let us now explain how we can derive Lemma~\ref{lem:2col} now follows from these claims. First note that, by symmetry, all these claims hold with $A_1$ replaced by $A_2$. By Claim~\ref{claim:1}, Lemma~\ref{lem:2col} holds if $\chi(A_1\cup\{2,3,5,6\})=2$ or $\chi(A_2\cup\{2,3,5,6\})=2$. So from now on we assume that $\chi(A_1\cup\{2,3,5,6\})>2$ and $\chi(A_2\cup\{2,3,5,6\})>2$.
Observe that if $\chi(A_1\cup\{0,1,2,3\})=\chi(A_2\cup\{0,4,5,6\})=2$, then Lemma~\ref{lem:2col} holds with $B_1=\{1,2,3\}$. Therefore, by symmetry, Claim~\ref{claim:2} leads to two cases: 
\begin{itemize}

\item $\chi(A_1\cup\{0,1,2,3\})=\chi(A_2\cup\{0,1,2,3\})=2$. In that case, Lemma~\ref{lem:2col} holds with $B_i=\{0,1,2,3\}$ unless $\chi(A_i\cup\{4,5,6\})>2$ for all $i \in \{1,2\}$. In that case, we can apply Claim~\ref{claim:3} to both $A_1$ and $A_2$, and Lemma~\ref{lem:2col} holds with $B_1=\{0,2,4\}$.

\item $\chi(A_1\cup\{0,4,5,6\})=\chi(A_2\cup\{0,4,5,6\})=2$.
In that case, Lemma~\ref{lem:2col} holds with $B_i=\{1,2,3\}$ unless $\chi(A_i\cup\{1,2,3\})>2$  for all $i \in \{1,2\}$. In that case, we can apply Claim~\ref{claim:4} to both $A_1$ and $A_2$, and Lemma~\ref{lem:2col} holds with $B_1=\{0,1,6\}$.
\end{itemize}

It remains to prove the four claims. First recall that $A_1\cup B$ contains a copy of $\mathrm{Pal}_{11}$, and the missing vertex lies in $A_1$. We can thus write $A_1=\{a,b,c,d,x\}$ where $B\cup\{a,b,c,d\}$ induces a copy of $\mathrm{Pal}_{11}$ and $a,b,c,d$ are in transitive order (see Figure~\ref{fig:situation174}). Observe also that (up to automorphism), $\mathrm{Pal}_{11}$ contains a unique copy of $TT_4$. Moreover, $W_1$ has no automorphism. Therefore, there is a unique way to put the arcs between $B$ and $\{a,b,c,d\}$, depicted in Figure~\ref{fig:situation174}). Each claim thus boils down to show that we can find a $2$-coloring of the right subgraph regardless of the neighborhoods of $x$. 

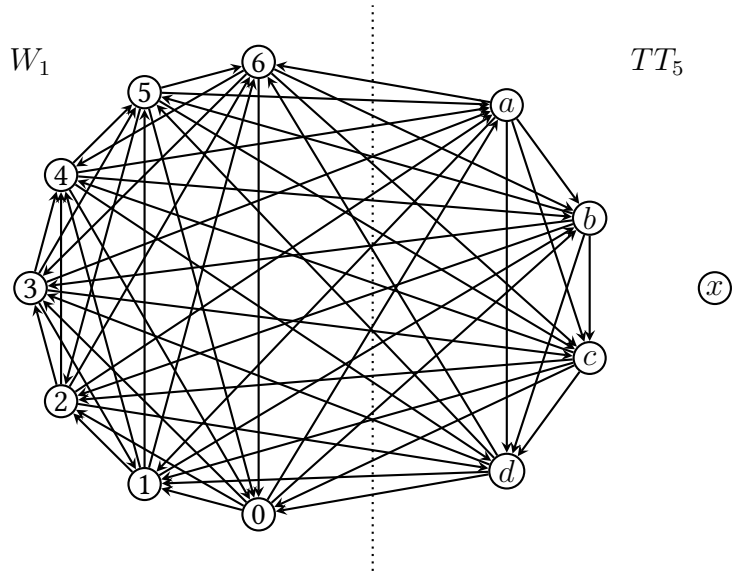
\begin{figure}[!ht]
\centering
\begin{tikzpicture}[thick,every node/.style={draw,circle,inner sep=1pt,minimum size=12pt},scale=1.5]
\draw[dotted] (0,-2.5) -- (0,2.5);
\node[color=white] at (2.5,2) {\textcolor{black}{$TT_5$}};
\node[color=white] at (-3,2) {\textcolor{black}{$W_1$}};
\tikzset{xshift=-1cm}
\node (6) at (90:2) {6};
\node (5) at (120:2) {5};
\node (4) at (150:2) {4};
\node (3) at (180:2) {3};
\node (2) at (210:2) {2};
\node (1) at (240:2) {1};
\node (0) at (270:2) {0};
\tikzset{xshift=1cm}
\node (a) at (54:2) {$a$};
\node (b) at (18:2) {$b$};
\node (c) at (-18:2) {$c$};
\node (d) at (-54:2) {$d$};
\node (x) at (0:3) {$x$};
\draw[-{stealth}] (0) -> (1);
\draw[-{stealth}] (0) to (2);
\draw[-{stealth}] (0) -> (3);
\draw[-{stealth}] (4) -> (0);
\draw[-{stealth}] (5) to (0);
\draw[-{stealth}] (6) -> (0);
\draw[-{stealth}] (1) -> (2);
\draw[-{stealth}] (2) -> (3);
\draw[-{stealth}] (3) -> (1);
\draw[-{stealth}] (4) -> (5);
\draw[-{stealth}] (5) -> (6);
\draw[-{stealth}] (6) -> (4);
\draw[-{stealth}] (1) -> (5);
\draw[-{stealth}] (1) -> (6);
\draw[-{stealth}] (1) -> (4);
\draw[-{stealth}] (2) -> (6);
\draw[-{stealth}] (2) -> (4);
\draw[-{stealth}] (5) -> (2);
\draw[-{stealth}] (3) -> (4);
\draw[-{stealth}] (3) -> (5);
\draw[-{stealth}] (6) -> (3);
\draw[-{stealth}] (a) -> (b);
\draw[-{stealth}] (a) -> (c);
\draw[-{stealth}] (a) -> (d);
\draw[-{stealth}] (b) -> (c);
\draw[-{stealth}] (b) -> (d);
\draw[-{stealth}] (c) -> (d);
\draw[-{stealth}] (2) -> (a);
\draw[-{stealth}] (b) -> (2);
\draw[-{stealth}] (c) -> (2);
\draw[-{stealth}] (2) -> (d);
\draw[-{stealth}] (3) -> (a);
\draw[-{stealth}] (b) -> (3);
\draw[-{stealth}] (3) -> (c);
\draw[-{stealth}] (d) -> (3);
\draw[-{stealth}] (4) -> (a);
\draw[-{stealth}] (4) -> (b);
\draw[-{stealth}] (c) -> (4);
\draw[-{stealth}] (4) -> (d);
\draw[-{stealth}] (5) -> (a);
\draw[-{stealth}] (b) -> (5);
\draw[-{stealth}] (5) -> (c);
\draw[-{stealth}] (d) -> (5);
\draw[-{stealth}] (0) -> (a);
\draw[-{stealth}] (0) -> (b);
\draw[-{stealth}] (c) -> (0);
\draw[-{stealth}] (d) -> (0);
\draw[-{stealth}] (a) -> (1);
\draw[-{stealth}] (1) -> (b);
\draw[-{stealth}] (c) -> (1);
\draw[-{stealth}] (d) -> (1);
\draw[-{stealth}] (a) -> (6);
\draw[-{stealth}] (6) -> (b);
\draw[-{stealth}] (6) -> (c);
\draw[-{stealth}] (d) -> (6);
\end{tikzpicture}
\caption{The known arcs in $A_1\cup B$.}
\label{fig:situation174}
\end{figure}

\begin{proof}[Proof of Claim~\ref{claim:1}]
We separate five cases depending on the rank of $x$ in the transitive order among $\{a,b,c,d,x\}$. 
\begin{itemize}
\item If $x<a$, then $\{x,4,a,b,d\}$ induces a $TT_5$, with either $x$ or $4$ as source depending on the orientation of the arc $4x$. To improve readability, we will present the vertices of the upcoming transitive tournaments in order within their set, using parenthesis when some vertices might be flipped depending on the orientation of the arc between them. In particular, the above copy of $TT_5$ will be written $\{(x,4),a,b,d\}$. Together with $\{c,0,1\}$ it gives a $2$-coloring of $A_1\cup\{0,1,4\}$.
\item If $a<x<b$, then $(\{c,4,d,0\},\{a,(1,x),b\})$ is a $2$-coloring of $A_1\cup\{0,1,4\}$.
\item If $b<x<c$, then either $(\{4,a,b,x,d\},\{c,0,1\})$ or $(\{0,a,1,b\},\{x,c,4,d\})$ is a $2$-coloring of $A_1\cup\{0,1,4\}$ (depending on the arc between $4$ and $x$).
\item If $c<x<d$, then $(\{0,a,1,b\},\{c,(4,x),d\})$ is a $2$-coloring of $A_1\cup\{0,1,4\}$.
\item If $d<x$, then $(\{a,c,d,(1,x)\},\{4,0,b\})$ is a $2$-coloring of $A_1\cup\{0,1,4\}$.\qedhere
\end{itemize}
\end{proof}

\begin{proof}[Proof of Claim~\ref{claim:2}]
Observe that $\gamma_1=(\{a,c,d,1\},\{0,b,2,3\})$ and $\gamma_2=(\{b,c,2,d\},\{0,3,a,1\})$ are two $2$-colorings of $\{a,b,c,d,0,1,2,3\}$, and that $\gamma_3=(\{a,6,b,c\},\{4,d,5,0\})$ and $\gamma_4=(\{4,a,b,d\},\{5,6,c,0\})$ are two $2$-colorings of $\{a,b,c,d,0,4,5,6\}$. We separate five cases depending on the rank of $x$ in the transitive order among $\{a,b,c,d,x\}$. In each case, we look for an extension of $\gamma_1$ or $\gamma_2$ into a $2$-coloring of $A_1\cup\{0,1,2,3\}$ or an extension of $\gamma_3,\gamma_4$ to $A_2\cup\{0,4,5,6\}$.
\begin{itemize}
\item If $x<a$, then we can extend $\gamma_4$.
\item If $a<x<b$, then we can extend $\gamma_3$.
\item If $b<x<c$, then we must have the arc $2x$ (resp. $x6$) for otherwise we can extend $\gamma_2$ (resp. $\gamma_3$). Now if there is an arc $0x$, we can extend $\gamma_1$, otherwise there is an arc $x0$ and we can extend $\gamma_4$.
\item If $c<x<d$, then we can extend $\gamma_2$.
\item If $d<x$, then we can extend $\gamma_1$.\qedhere
\end{itemize}
\end{proof}

\begin{proof}[Proof of Claim~\ref{claim:3}]
We know that $b<x$ and $x6$ is an arc otherwise $(\{x,4,a,b,d\},\{5,6,c\})$ or $(\{a,6,x,b,c\},\{4,d,5\})$ is a $2$-coloring of $A_1\cup\{4,5,6\}$, which is not possible by hypothesis. Moreover, $x2$ is an arc otherwise $(\{2,a,x,d,6\},\{b,3,5,c\})$ is a $2$-coloring of $A_1\cup \{2,3,5,6\}$. Finally, we have $d<x$, otherwise either $(\{5,a,6,c\},\{b,x,2,d,3\})$ or $(\{2,a,d,6\},\{b,3,x,5,6\})$ is a $2$-coloring of $A_1\cup\{2,3,5,6\}$ (depending on the arc $x3$). 

Now $(\{4,a,b,d\},\{c,0,x,2\})$ is a $2$-coloring of $A_1\cup\{0,2,4\}$ and $(\{a,d,1,x,6\},\{b,3,5,c\})$ is a $2$-coloring of $A_1\cup\{1,3,5,6\}$.
\end{proof}

\begin{proof}[Proof of Claim~\ref{claim:4}]
We know that $x<b$ and $1x,2x$ are arcs otherwise $(\{a,c,d,1,x\},\{b,2,3\})$ or $(\{3,a,c,1\},\{b,2,x,d\})$ is a $2$-coloring of $A_1\cup\{1,2,3\})$ (depending on the orientation of $xd$). Moreover, $6x$ is an arc otherwise $(\{2,a,x,d,6\},\{b,3,5,c\})$ is a $2$-coloring of $A_1\cup\{2,3,5,6\}$.

Assume that $a<x<b$. Then we have the arc $x3$ otherwise $(\{b,c,2,d\},\{3,a,1,x\})$ is a $2$-coloring of $A_1\cup\{1,2,3\}$. But this is impossible since $(\{2,a,d,6\},\{x,b,3,5,c\})$ or $(\{5,a,6,x,c\},\{b,2,d,3\})$ is a $2$-coloring of $A_1\cup\{2,3,5,6\}$ (depending on the arc $x5$). 

Therefore we have $x<a$, so $(\{a,c,d,1\},\{6,0,x,b\})$ is a $2$-coloring of $A_1\cup\{0,1,6\}$ and $(\{2,4,x,a,d\},\{b,3,5,c\})$ is a $2$-coloring of $A_1\cup\{2,3,4,5\}$. 
\end{proof}

\section{A 5-chromatic tournament on 19 vertices}
\label{sec:19}
In his seminal paper~\cite{NL-tournaments}, Neumann-Lara said that there exists a $5$-chromatic tournament on $19$ vertices but gave no details on the structure of this tournament or the proof of this fact. 
He actually explained how to construct such a tournament six years later in~\cite{neumann2000dichromatic}, as an illustration of his results on the Zykov sums of digraphs.
Before we knew about this paper, we looked for a 5-chromatic tournament and found independently the same tournament.
For the sake of completeness, we present this tournament in this section and prove its $5$-chromaticity (the construction is outlined in Figure~\ref{fig:19}). 
We leave the existence of another 5-chromatic tournament on 19 vertices as an important open question.

\begin{theorem}\label{thm:19vertices}
There is a $5$-chromatic tournament on $19$ vertices.
\end{theorem}

\begin{proof}
We consider the tournament $\mathrm{Pal}_7$ with vertices labeled as in Figure~\ref{fig:T7}.
Let $D$ be the 19-vertex tournament obtained from $\mathrm{Pal}_7$ by blowing-up every vertex besides $0$ into a triangle (see Figure~\ref{fig:19}). More precisely, every vertex $i \in \{1,\ldots,6\}$ is replaced by three vertices $i_1,i_2,i_3$ inducing a triangle and $i_jk_\ell$ is an arc of $D$ if and only if $ik$ is an arc of $\mathrm{Pal}_7$ and $i_j0$ is an arc of $D$ if and only if $i0$ is an arc of $\mathrm{Pal}_7$.

For the sake of a contradiction, suppose that $D$ admits a $4$-coloring. 
In particular, the vertices of the each triangle $i_1,i_2,i_3$ receive at least two different colors.
We consider the multicoloring of $\mathrm{Pal}_7$ naturally associated with the coloring of $D$
(that is, each vertex $i \in \{1,\ldots,6\}$ is given the colors of $i_1,i_2,i_3$)
and note the following.
\begin{itemize}
	\item Vertex $0$ is colored with one color.        
    \item Every vertex $1, \dots, 6$ is colored with at least two colors.
    \item Every color class induces a transitive tournament.
\end{itemize}

It follows that there are at least $1 + 6\times 2= 13$ associations of colors to vertices of $\mathrm{Pal}_7$. Therefore, some color appears on at least $\lceil \frac{13}{4}\rceil = 4$ vertices of $\mathrm{Pal}_7$, which is a contradiction since $\mathrm{Pal}_7$ is $TT_4$-free.
\end{proof}

For $n$ up to 4, we actually know all the smallest $n$-chromatic tournaments, and this knowledge has been crucial for determining $n_5$. Thus, we believe it could be important to determine how many 5-chromatic tournaments there are on 19 vertices.
We remark that for every tournament on $7$ vertices distinct from $\mathrm{Pal}_7$, the construction of blowing-up all vertices into triangles results in a $4$-colorable tournament
(on $21$ vertices).
Also, note that any small modification of $D$ such as reverting or removing an arbitrary arc would make it $4$-colorable too.

If $D$ is the unique 5-chromatic tournament on 19 vertices, it would imply that removing any $TT_5$ from a $6$-chromatic tournament on $24$ vertices yields $D$, and that any $6$-chromatic tournament on $25$ vertices is either $TT_6$-free or is obtained by gluing $TT_6$ with $D$. However, while such a result would be a nice step forward, the methods from Sections 4 and 6 still do not seem powerful enough in their current state to tackle the case of $6$-chromatic tournaments.

\begin{figure}[!ht]
\centering
\begin{tikzpicture}[thick,v/.style={minimum size = 10pt, circle}]
\node[v,draw,outer sep=4pt] (0) at (90:4) {};
\node[v] (1) at (142:4) {\tri};
\node[v] (2) at (193:4) {\tri};
\node[v] (3) at (245:4) {\tri};
\node[v] (4) at (296:4) {\tri};
\node[v] (5) at (348:4) {\tri};
\node[v] (6) at (39:4) {\tri};
\node at (0,0) {$\mathrm{Pal}_7$};
\draw[ultra thick, ->] (0) to (1);
\draw[ultra thick,->] (0) to (2);
\draw[ultra thick,->] (0) to (4);
\draw[ultra thick,->] (1) to (2);
\draw[ultra thick,->] (1) to (3);
\draw[ultra thick,->] (1) to (5);
\draw[ultra thick,->] (2) to (3);
\draw[ultra thick,->] (2) to (4);
\draw[ultra thick,->] (2) to (6);
\draw[ultra thick,->] (3) to (4);
\draw[ultra thick,->] (3) to (5);
\draw[ultra thick,->] (3) to (0);
\draw[ultra thick,->] (4) to (5);
\draw[ultra thick,->] (4) to (6);
\draw[ultra thick,->] (4) to (1);
\draw[ultra thick,->] (5) to (6);
\draw[ultra thick,->] (5) to (0);
\draw[ultra thick,->] (5) to (2);
\draw[ultra thick,->] (6) to (0);
\draw[ultra thick,->] (6) to (1);
\draw[ultra thick,->] (6) to (3);
\end{tikzpicture}
\caption{The construction of a $5$-chromatic tournament on $19$ vertices. Each thick arrow represents a set of arcs with the same orientation. Thick arrows outline the structure of $\mathrm{Pal}_7$}
\label{fig:19}
\end{figure}
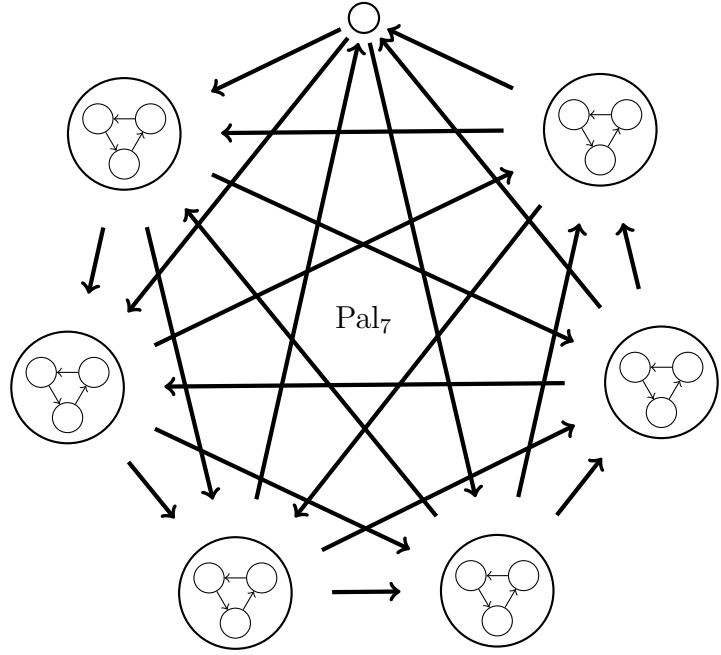

\section{The answer to Neumann-Lara's question}
\label{sec:18}

In Sections~\ref{sec:17} and~\ref{sec:19},
we showed that all tournaments of order $17$ are $4$-colorable
and that there is a $5$-chromatic tournament of order $19$.
In this section, we consider tournaments on $18$ vertices
and we answer the question on the order of a smallest $5$-chromatic tournament.
We present a computer proof of the following.

\begin{theorem}
\label{thm:18}
Every tournament on $18$ vertices is $4$-colorable.
\end{theorem}

We cannot simply check every tournament on $18$ vertices
since there are more than $10^{30}$ such tournaments~\cite{oeis}.
We first need to reduce the problem to be able to solve it even with the help of a computer. 
To this end, we consider a hypothetical counterexample to Theorem~\ref{thm:18}
and investigate its structural properties.

We let $T$ be a $5$-chromatic tournament of order $18$,
and we apply a similar but more involved reasoning as in Sections~\ref{sec:12} and~\ref{sec:17}.
First, we show that $T$ contains at least two disjoint copies of $TT_5$ in Subsection~\ref{subsec:61}.
In Subsections~\ref{subsec:62} and~\ref{subsec:63}, we then discuss two cases based on the number of disjoint copies of $TT_5$ in $T$.
For each case, we present an algorithm which leads to a contradiction with the choice of $T$.

\subsection{$T$ must contain two disjoint $TT_5$}
\label{subsec:61}

This section is devoted to proving the following generalization of Lemma~\ref{lem:2TT5in17}.

\begin{lemma}
\label{lem:2TT5in18}
The tournament $T$ contains at least two disjoint copies of $TT_5$.
\end{lemma}

Using Lemma~\ref{lem:sub}, $T$ must contain a copy of $TT_5$. Fix an arbitrary such copy and consider the tournament $T'$ induced by the remaining $13$ vertices. If $T'$ is not $TT_5$-free, then the lemma follows. So we can assume that $T'$ is $TT_5$-free. In particular, $T'$ is isomorphic to $X_{13}$ by Theorem~\ref{thm:X13}.

Let $a_0<\cdots<a_4$ be the vertices of $T\setminus T'$ in transitive order. Note that $T'+a_0$ has 14 vertices, hence contains a $TT_5$ by Theorem~\ref{thm:14}. Since $T'$ is $TT_5$-free, this $TT_5$ contains $a_0$. Let us denote by $\{a_0,a'_1,a'_2,a'_3,a'_4\}$ the vertices of this $TT_5$, where $a'_1<\cdots<a'_4$. If $T'\setminus\{a'_1,a'_2,a'_3,a'_4\} \cup\{a_1,a_2,a_3,a_4\}$ contains $TT_5$ then we are done. Otherwise, this tournament is also isomorphic to $X_{13}$. 

By Proposition~\ref{clm:struct_X13}, there is an isomorphism $f:  T\setminus\{a'_1,a'_2,a'_3,a'_4\}\to X_{13}$ such that $f(a_1)=0$ and $f(a_2)\in\{1,2\}$. Moreover, by Proposition~\ref{claim:main6.1}, we may even assume that the quadruple $(f(a_1),f(a_2),f(a_3),f(a_4))$ is either $(0,1,2,3)$, $(0,1,3,6)$, $(0,1,6,2)$ or $(0,2,3,5)$. Similarly, we may define an isomorphism $f':T'\to X_{13}$ where $(f'(a'_1),f'(a'_2),f'(a'_3),f'(a'_4))$ is one of these four quadruples. 

Using again Proposition~\ref{claim:main6.1}, we first get that  $T'\setminus\{a'_1,a'_2,a'_3,a'_4\}$ has no non-trivial automorphism, hence $f$ and $f'$ coincide on these 9 vertices, and moreover that $f(a_i)=f'(a'_i)$ for each $i\in[1,4]$. Let $j\in[1,4]$ such that $f(a_j)=\max \{f(a_1),f(a_2),f(a_3),f(a_4)\}$. Then $\{f^{-1}(10),f^{-1}(11),f^{-1}(12),a_1,a'_1\}$ and $\{a_j,a'_j, a'_j+1,a'_j+2,a'_j+3\}$ are two disjoint copies of $TT_5$ in $T$. This concludes the proof of Lemma~\ref{lem:2TT5in18}.

\subsection{Description of the program}
\label{subsec:62}
Let $A_1,A_2$ be two disjoint copies of $TT_5$ in $T$ and $B$ be the subtournament induced by the $8$ remaining vertices. Observe that $B$ is a $3$-chromatic tournament (otherwise $T$ would be $4$-colorable). Therefore, $B$ lies among a list of $258$ tournaments\footnote{Generated using \texttt{nauty}~\cite{MCKAY201494}, see the file~\href{https://github.com/tpierron/5chromatictournaments/blob/main/3-chromatic\%20tournaments\%20on\%208\%20vertices.ipynb}{$3$-chromatic tournaments on $8$ vertices.ipynb}.}. 

Note that the direction of $105$ arcs remains unfixed (the 25 arcs between $A_1$ and $A_2$ and then $80$ arcs between $B$ and $A_1\cup A_2$), hence an exhaustive search of all these tournaments is still unreasonable. Our method consists in using an approach similar to the one used in Section~\ref{sec:17}. More precisely, we would like to prove adapt Lemma~\ref{lem:2col} and prove that, for each choice of $B$, it is possible to split $B$ into $B_1\cup B_2$ such that $A_1\cup B_1$ and $A_2\cup B_2$ are both $2$-colorable. Unfortunately this method will fail for some choices of $B$ but will permit to restrict the number of cases to consider. 

To explain how our program works, we need some terminology. Fix a $3$-chromatic $TT_5$-free tournament $B$ on $8$ vertices. Let $A$ be a copy of $TT_5$ and let $C$ be a gluing of $A$ and $B$. We say that $C$ is a $13$-\emph{completion} of $B$ if $C$ is $4$-chromatic. Observe that, in $T$, $A_1\cup B$ and $A_2\cup B$ are two $13$-completions of $B$.

The \emph{type} of $C$ is the set of subtournaments $B'$ of $B$ such that $3\leqslant |B'|\leqslant 5$ and $A\cup B'$ induce a $2$-colorable tournament. We say that two types $\mathcal{T}_1,\mathcal{T}_2$ are \emph{compatible} if the vertices of $B$ can be partitioned as $B_1\cup B_2$ where $B_1\in \mathcal{T}_1$ and $B_2\in \mathcal{T}_2$. Since $T$ is $5$-chromatic, the completions $A_1\cup B$ and $A_2\cup B$ are not compatible.

Our program works as follows. For each possible choice of $B$, we generate all of its $13$-completions using Algorithm~\ref{algo:dsmash}. Then, we consider each pair of completions with incompatible types. For each such pair, we construct an $18$-vertex oriented graph by identifying the two copies of $B$. We then apply a slightly modified version of Algorithm~\ref{algo:dsmash} to check that orienting the $25$ missing arcs only yields $4$-colorable tournaments. 

For each of the $258$ choices for $B$, this algorithm's running time may take up to roughly ten days on a standard computer. This directly yields a parallel algorithm (using one core per choice of $B$), which concludes in several years of total computation time. In the following, we present a deeper analysis of the tournaments, which allows us to design a faster algorithm. 

To this end, we separate two new cases: either $T$ contains precisely two disjoint copies of $TT_5$, or it contains three of them. We handle the former case with the current approach, and the latter in the next subsection. This allows for a faster treatment since we only need to consider the $94$ cases where $B$ is $TT_5$-free. Moreover, it makes also Algorithm~\ref{algo:dsmash} run faster since we can cut branches as soon as we find two disjoint copies of $TT_5$. 

We provide a sequential implementation of this procedure in the file~\href{https://github.com/tpierron/5chromatictournaments/blob/main/section62.ml}{section62.ml}. For each oriented graph, our program takes between a few hours and a few days on a standard computer and outputs no tournament. Note that this part can again be easily parallelized. Therefore, we get the following.

\begin{lemma}
If there is a tournament $T$ on $18$ vertices that is $5$-chromatic then $T$ contains three pairwise disjoint copies of $TT_5$.
\end{lemma}

\subsection{$T$ has three disjoint $TT_5$}
\label{subsec:63}
Let $A_1,A_2,A_3$ be three disjoint copies of $TT_5$ in $T$, and $B$ be the set of the three remaining vertices (that must induce a directed triangle). Note that for each $i$, $A_i\cup B$ induces a $3$-chromatic tournament on $8$ vertices that contains a $TT_5$. Moreover, for every $i\neq j$, $A_i\cup A_j\cup B$ induces a $4$-chromatic tournament on $13$-vertices. 

Similarly to the previous section, our goal is to generate the candidates for $A_i\cup B$, then for $A_i\cup A_j\cup B$, and finally for $T$. We again rely on the notion of completion. An \emph{8-completion} is a $3$-chromatic tournament on $8$ vertices together with a fixed copy of $TT_5$ in it. Two $8$-completions are isomorphic if there is an isomorphism between the tournaments that fixes the distinguished copies of $TT_5$. Equivalently, this means that one can be obtained from the other by a circular permutation of the three vertices that are not in the distinguished $TT_5$ (which is a triangle).

There are $256$ non-isomorphic $8$-completions. Note that it is not surprising that this number is larger than the number of $3$-chromatic tournaments on $8$ vertices containing $TT_5$ since such a tournament may actually contain several copies of $TT_5$.

For every pair $(C,C')$ of $8$-completions,
we consider the $13$-vertex oriented graphs obtained by identifying
the three non-distinguished vertices of $C$ with the three non-distinguished vertices of $C'$.
More precisely, the set of the non-distinguished vertices induces a directed triangle in $C$ and in $C'$,
this yields three possible ways for making the identification, and hence we obtain three $13$-vertex oriented graphs.
For each of these $13$-vertex oriented graphs,
we use Algorithm~\ref{algo:dsmash} which outputs all arc-extensions to $4$-chromatic tournaments
(each tournament is obtained by adding $25$ arcs).
The resulting tournaments are candidates for $A_i\cup A_j\cup B$, and we call them $13$-completions of $(C,C')$. 

We may now generate all candidates for $T$. We consider three 8-completions $C_1,C_2,C_3$, then generate all $13$-completions $C_{12}$ (resp. $C_{13}, C_{23}$) of $(C_1,C_2)$ (resp. $(C_1,C_3),(C_2,C_3)$) again with Algorithm~\ref{algo:dsmash}. We finally construct a 18-vertex tournament by identifying the vertices inducing $C_1$ in $C_{12}$ and $C_{13}$, those inducing $C_2$ in $C_{12}$ and $C_{23}$ and those inducing $C_3$ in $C_{13}$ and $C_{23}$ (see Figure~\ref{fig:glue}). We then compute the chromatic number of each such $18$-vertex tournament. Each time, this chromatic number is $4$, which concludes the proof of Theorem~\ref{thm:18}. The corresponding program can be found in~\href{https://github.com/tpierron/5chromatictournaments/blob/main/section63.ml}{section63.ml}. We provide here a sequential implementation of this procedure, which runs in roughly six months on a standard computer. However, note that the computations of the $13$-completions can be parallelized (and so can be the $4$-colorability check for the $18$-vertex candidates), so the program could run much faster if several cores are available.

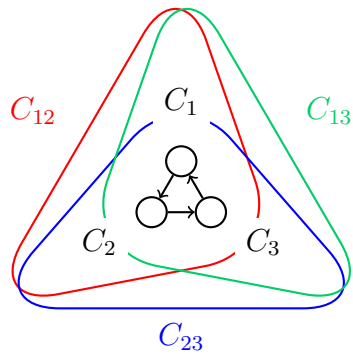
\begin{figure}[!ht]
\centering
\begin{tikzpicture}[thick, scale=1.5]
\draw [rounded corners=10mm,red] (-30:1)--(90:2)--(210:2)--cycle;
\node[red] at (150:1.5) {$C_{12}$};
\draw [rounded corners=10mm,blue] (90:1)--(330:2)--(210:2)--cycle;
\node[blue] at (270:1.25) {$C_{23}$};
\draw [rounded corners=10mm,ForestGreen] (-150:1)--(330:2)--(90:2)--cycle;
\node [ForestGreen] at (30:1.5) {$C_{13}$};
\node[draw, circle] (A) at (90:.3) {};
\node[draw, circle] (B) at (210:.3) {};
\node[draw, circle] (C) at (330:.3) {};
\draw[->] (A) to (B);
\draw[->] (B) to (C);
\draw[->] (C) to (A);
\node[fill=white] at (-30:.83) {$C_3$};
\node[fill=white] at (90:.83) {$C_1$};
\node[fill=white] at (-150:.83) {$C_2$};
\end{tikzpicture}
\caption{The structure of $T$.}\label{fig:glue}
\end{figure}

\bibliographystyle{alpha}

\begin{thebibliography}{ABHR22}

\bibitem[ABHR22]{4.67}
Pierre Aboulker, Thomas Bellitto, Frédéric Havet, and Clément Rambaud.
\newblock On the minimum number of arcs in $k$-dicritical oriented graphs.
\newblock {\em arXiv}, 2022.

\bibitem[BBSS20]{dicritical}
J{\o}rgen Bang{-}Jensen, Thomas Bellitto, Thomas Schweser, and Michael
  Stiebitz.
\newblock Haj{\'{o}}s and {O}re constructions for digraphs.
\newblock {\em Electronic Journal of Combinatorics}, 27(1):P1.63, 2020.

\bibitem[Chv70]{Chvatal}
V\'{a}clav Chv\'{a}tal.
\newblock The smallest triangle-free 4-chromatic 4-regular graph.
\newblock {\em Journal of Combinatorial Theory}, 9:93--94, 07 1970.

\bibitem[Cou94]{courcelle}
Bruno Courcelle.
\newblock The monadic second order logic of graphs {VI:} {O}n several
  representations of graphs by relational structures.
\newblock {\em Discrete Applied Mathematics}, 54(2-3):117--149, 1994.

\bibitem[EM64]{ErdosMoser}
Paul Erd{\H o}s and Leo Moser.
\newblock On the representation of directed graphs as unions of orderings.
\newblock {\em Mathematical Institute of the Hungarian Academy of Sciences},
  9:125--132, 1964.

\bibitem[Erd79]{ErdosNL1}
Paul Erd{\H{o}}s.
\newblock Problems and results in number theory and graph theory.
\newblock In {\em Proceedings of the 9th Manitoba Conference on Numerical
  Mathematics and Computing}, pages 3--21, 1979.

\bibitem[Goe20]{Goedgebeur}
Jan Goedgebeur.
\newblock On minimal triangle-free 6-chromatic graphs.
\newblock {\em Journal of Graph Theory}, 93(1):34--48, 2020.

\bibitem[HK15]{Hoshino}
Richard Hoshino and Ken{-}ichi Kawarabayashi.
\newblock The edge density of critical digraphs.
\newblock {\em Combinatorica}, 35(5):619--631, 2015.

\bibitem[JR95]{JensenRoyle}
Tommy~R. Jensen and Gordon~F. Royle.
\newblock Small graphs with chromatic number 5: {A} computer search.
\newblock {\em Journal of Graph Theory}, 19(1):107--116, 1995.

\bibitem[KS20]{Kostochka}
Alexandr~V. Kostochka and Michael Stiebitz.
\newblock The minimum number of edges in 4-critical digraphs of given order.
\newblock {\em Graphs and Combinatorics}, 36(3):703--718, 2020.

\bibitem[MP14]{MCKAY201494}
Brendan~D. McKay and Adolfo Piperno.
\newblock Practical graph isomorphism, {II}.
\newblock {\em Journal of Symbolic Computation}, 60:94--112, 2014.

\bibitem[Neu82]{def_dicoloring}
Victor Neumann{-}Lara.
\newblock The dichromatic number of a digraph.
\newblock {\em Journal of Combinatorial Theory, Series {B}}, 33(3):265--270,
  1982.

\bibitem[Neu94]{NL-tournaments}
Victor Neumann{-}Lara.
\newblock The 3 and 4-dichromatic tournaments of minimum order.
\newblock {\em Discrete Mathematics}, 135(1-3):233--243, 1994.

\bibitem[NL00]{neumann2000dichromatic}
V{\'\i}ctor Neumann-Lara.
\newblock Dichromatic number, circulant tournaments and zykov sums of digraphs.
\newblock {\em Discussiones Mathematicae Graph Theory}, 20(2):197--207, 2000.

\bibitem[RP70]{REID1970225}
Kenneth~B. Reid and Ernest~T. Parker.
\newblock Disproof of a conjecture of {E}rd{\H{o}}s and {M}oser on tournaments.
\newblock {\em Journal of Combinatorial Theory}, 9(3):225--238, 1970.

\bibitem[SF98]{sanchez1998tournaments}
Adolfo Sanchez-Flores.
\newblock On tournaments free of large transitive subtournaments.
\newblock {\em Graphs and Combinatorics}, 14(2):181--200, 1998.

\bibitem[SI]{oeis}
Neil J.~A. Sloane and The OEIS~Foundation Inc.
\newblock The on-line encyclopedia of integer sequences.
\newblock \url{https://oeis.org/A000568}.

\bibitem[Ste59]{stearns1959voting}
Richard Stearns.
\newblock The voting problem.
\newblock {\em The American Mathematical Monthly}, 66(9):761--763, 1959.

\end{thebibliography}

\end{document}